\documentclass[11pt]{article}

\usepackage{amsmath,amssymb,amsthm}
\usepackage{color,url}

\usepackage[top=1in, bottom=1in, left=1.25in, right=1.25in]{geometry}
\definecolor{darkred}{rgb}{0.5,0,0}
\definecolor{darkgreen}{rgb}{0,0.35,0}
\definecolor{darkblue}{rgb}{0,0,0.55}

\usepackage[pdfstartview=FitH,pdfpagemode=UseNone,colorlinks,linkcolor=darkblue,filecolor=darkred,citecolor=darkgreen,urlcolor=blue,pagebackref]{hyperref}

\usepackage[capitalise,nameinlink]{cleveref}
\usepackage{mathpazo}

\newtheorem{theorem}{Theorem}[section]
\newtheorem{lemma}[theorem]{Lemma}
\newtheorem{proposition}[theorem]{Proposition}
\newtheorem{claim}[theorem]{Claim}

\newcommand\F{\mathbb{F}}
\newcommand{\C}{\mathcal{C}}
\newcommand\one{\mathbf{1}}
\newcommand\dho{\partial}
\newcommand\set[1]{{\{{#1}\}}}
\newcommand\ip[1]{{\left\langle{#1}\right\rangle}}
\providecommand{\B}{\mathcal{B}}
\DeclareMathOperator{\Syst}{CoSys}
\DeclareMathOperator{\dist}{dist}

\DeclareMathOperator{\im}{im}
\renewcommand{\epsilon}{\varepsilon}
\newcommand{\eps}{\epsilon}

\newcommand{\ball}{\mathsf{N}}
\newcommand{\ballX}{\ball}
\newcommand{\ballB}{\hat{\ball}}
\newcommand{\sqbinom}{\genfrac{[}{]}{0pt}{}}

\def\showauthornotes{0}
\ifnum\showauthornotes=1
\newcommand{\Authornote}[2]{{\sf\small\color{red}{[#1: #2]}}}
\newcommand{\Authorcomment}[2]{{\sf \small\color{gray}{[#1: #2]}}}
\newcommand{\Authorfnote}[2]{\footnote{\color{red}{#1: #2}}}
\else
\newcommand{\Authornote}[2]{}
\newcommand{\Authorcomment}[2]{}
\newcommand{\Authorfnote}[2]{}
\fi

\newcommand{\remove}[1]{}
\newcommand{\norm}[1]{\ensuremath{\left\lVert #1 \right\rVert}}
\newcommand{\abs}[1]{\ensuremath{\left\lvert #1 \right\rvert}}
\newcommand{\U}[1]{\mathbf{u}_{#1}}
\newcommand{\Uempty}{\mathbf{u}_{\emptyset}}

\title{Explicit SoS lower bounds from high-dimensional expanders}

\author{Irit Dinur\thanks{Weizmann Institute of Science, ISRAEL. email: {\tt irit.dinur@weizmann.ac.il}. Research supported by ERC-CoG grant number 772839.}
  \and
  Yuval Filmus\thanks{Technion — Israel Institute of Technology, ISRAEL. email: {\tt yuvalfi@cs.technion.ac.il}. This project has received funding from the European Union's Horizon 2020 research and innovation programme under grant agreement No~802020-ERC-HARMONIC.}
  \and
  Prahladh Harsha\thanks{Tata Institute of Fundamental Research, INDIA. email: {\tt prahladh@tifr.res.in}. Research supported by the Department of Atomic Energy, Government of India, under project no. RTI4001 and in part by the Swarnajayanti fellowship.}
  \and
  Madhur Tulsiani\thanks{Toyota Technological Institute at Chicago. email: {\tt madhurt@ttic.edu}. Research supported by NSF grant CCF-1816372.} 
}

\date{}

\begin{document}
\maketitle

\begin{abstract}
We construct an explicit family of 3XOR instances which is hard for $O(\sqrt{\log n})$ levels of the Sum-of-Squares hierarchy. In contrast to earlier constructions, which involve a random component, our systems can be constructed explicitly in deterministic polynomial time.

Our construction is based on the high-dimensional expanders devised by Lubotzky, Samuels and Vishne, known as LSV complexes or Ramanujan complexes, and our analysis is based on two notions of expansion for these complexes: cosystolic expansion, and a local isoperimetric inequality due to Gromov.

Our construction offers an interesting contrast to the recent work of Alev, Jeronimo and the last author~(FOCS 2019). They showed that 3XOR instances in which the variables correspond to vertices in a high-dimensional expander are easy to solve. In contrast, in our instances the variables correspond to the edges of the complex.
\end{abstract}

\section{Introduction}

We describe a new family of instances of 3XOR, based on high-dimensional expanders, that are hard for the Sum-of-Squares (SoS) hierarchy of semidefinite programming relaxations, which is the most powerful algorithmic framework known for optimizing over constraint satisfaction problems. Unlike previous constructions of 3XOR hard instances for SoS, our construction is explicit, as it is based on the explicit construction of high-dimensional expanders due to Lubotzky, Samuels and Vishne~\cite{LubotzkySV2005-exphdx,LubotzkySV2005-hdx}, which we refer to henceforth as LSV complexes.
\vspace{-2 pt}
\begin{theorem} \label{thm:main-intro}
  There exists a constant $\mu \in (0,1)$ and an infinite family of 3XOR instances on $n$ variables, constructible in deterministic polynomial time, satisfying the following:
  \begin{itemize}
  \item No assignment satisfies more than $1-\mu$ fraction of the constraints.
  \vspace{-3 pt}
  \item Relaxations obtained by $O(\sqrt{\log n})$ levels of the SoS hierarchy fail to refute the instances.
  \end{itemize}
\end{theorem}
We also remark that our construction can be used to obtain explicit integrality gap instances for various other optimization problems, using reductions in the SoS hierarchy~\cite{Tulsiani2009}. In particular, while our instances on the LSV complexes exhibit an integrality gap of $1-\mu$ vs.\ $1$ for the SoS hierarchy, reductions can be used to obtain explicit $3$XOR instances with a gap of $1/2+\epsilon$ vs.\ $1-\epsilon$ for any $\epsilon > 0$. Indeed, this yields explicit hard instances with optimal gaps for all approximation resistant predicates based on pairwise independent subgroups~\cite{Chan2016}.

\paragraph{Structured instances from High-dimensional expanders.}
High-dimensional expanders (HDXs) are a high-dimensional analog of expander graphs.
In recent years they have found a variety of applications in theoretical computer science, such as efficient CSP optimization~\cite{AlevJT2019}, improved sampling algorithms~\cite{AnariLGV2019,AnariLG2020,AlevL2020}, quantum LDPC codes~\cite{EvraKZ2020,KaufmanT2020}, novel lattice constructions~\cite{KaufmanM2018}, direct sum testing~\cite{GotlibK2019}, and others. Explicit constructions of HDXs have also led to improved list-decoding algorithms~\cite{DinurHKNT2019,AlevJQST2020} and to sparser agreement tests~\cite{DinurK2017,DiksteinD2019}. In this work, we show how these explicit constructions can be used to construct explicit hard instances for SoS.

High-dimensional expanders are bounded-degree (hyper)graphs (or rather, simplicial complexes) with certain expansion properties. A simplicial complex is a non-empty collection of down-closed sets. Given a simplicial complex $X$, we will refer by $X(i)$ the family of all $i$-dimensional sets in $X$ (i.e., sets of size $i+1$). The dimension of the simplicial complex $X$ is the maximal dimension of any set in it. It will be convenient to refer to the sets of dimension 0, 1, 2, 3 as vertices, edges, triangles, tetrahedra, respectively. Thus, a graph $G=(V,E)$ is a 1-dimensional complex, while in this work we will be using complexes of dimension at least 2. Given a 2-dimensional complex $X= (X(0), X(1),X(2))$, there are two natural ways to construct a 3XOR instance based on $X$ --- a vertex-variable construction and an edge-variable construction.
Let $\beta \colon X(2) \to \F_2$ be any $\F_2$-valued function on the set $X(2)$ of triangles.
\begin{description}
\item[Vertex-variable construction:]
  The 3XOR instance corresponding to $(X,\beta)$ consists of the following constraints: $x_u + x_v + x_w = \beta_{\{u,v,w\}}$ for each $\{u,v,w\} \in X(2)$.
\item[Edge-variable construction:]
  The 3XOR instance corresponding to $(X,\beta)$ consists of the following constraints: $x_{\{u,v\}} + x_{\{v,w\}} + x_{\{w,u\}} = \beta_{\{u,v,w\}}$ for each $\{u,v,w\} \in X(2)$.
\end{description}

The vertex-variable construction whose underlying structure is a high-dimensional expander has been studied by Alev, Jeronimo and the last author~\cite{AlevJT2019}. They gave an efficient algorithm for approximating vertex-variable constraint satisfaction problems (not necessarily 3XOR) on an underlying high-dimensional expander. Their result is a generalization to higher dimensions of the corresponding result for graphs that ``CSPs are easy on expanders''~\cite{BarakRS2011,GuruswamiS2011}. They prove this by showing that certain types of random walks on vertices converge very fast on high-dimensional expanders. However, the same analysis fails to show a similar result for the edge-variable construction, as the corresponding random walk on edges of a high-dimensional expander does not mix. Our work shows that this difference isn't just a technical limitation of their analysis; it is inherent. The edge-variable variant is truly hard, at least for SoS. This demonstrates an interesting subtlety in the structure of high-dimensional expanders, and how it relates to optimization.

To understand our edge-variable construction better, it will be convenient to set up some notation. Let $C^i$ denote the set of all $\F_2$-valued functions on $X(i)$. For each $0 \leq i < d$, consider the operator $\delta_i \colon C^i \to C^{i+1}$ defined as follows:
\[
 \delta_if(s) := \sum_{u\in s} f(s-\set u).
\]
This is usually referred to as the coboundary operator. Let $B^i$ be the image of $\delta_{i-1}$, and let $Z^i$ be the kernel of $\delta_i$. Clearly, $B^i, Z^i \subseteq C^i$. Furthermore, it is not hard to see that $B^i \subseteq Z^i$. It easily follows from the definitions that the edge-variable construction corresponding to $(X,\beta)$ is a satisfiable instance iff $\beta \in B^2$.



Typically, soundness of SoS-hard instances is proved by choosing $\beta$ at random from $C^2$. In contrast, we construct our explicit instances by choosing the function $\beta$ more carefully, and relying on a certain type of expansion property of the complex. Recall that $B^2 \subset Z^2$, and the instance is satisfiable iff $\beta \in B^2$. Complexes for which $B^2 = Z^2$ are said to have trivial second cohomology. We will be working with complexes with non-trivial second cohomology, i.e., $B^2 \neq Z^2$. This lets us choose a $\beta \in Z^2 \setminus B^2$ to prove soundness. It is known that the explicit constructions of HDXs due to Lubotzky, Samuels and Vishne~\cite{LubotzkySV2005-exphdx,LubotzkySV2005-hdx} have non-trivial second cohomology.\footnote{More accurately, their construction depends on the group defining the quotient. They show that a certain choice of groups yields non-trivial second cohomology.} In fact, these complexes have the stronger property (due to a theorem of Evra and Kaufman~\cite{EvraK2016}) that all $\beta \in Z^2 \setminus B^2$ are not only not in $B^2$, but in fact far from any function in $B^2$. This latter property follows from the cosystolic expansion of the complex, and forms the basis for the soundness of our instances.

How do we prove the completeness of our instance, namely, that SoS fails to detect that it is a negative instance? The LSV construction is a quotient of the so-called affine building which is, from a topological point of view, a simple ``Euclidean-like'' object with trivial cohomologies. The hardness of our instance comes from the inherent difference between the LSV complex and the building, which cannot be seen through local balls whose radius is at most the injectivity radius of the complex, in our case $\Theta(\log n)$. Locally, the LSV quotient is isomorphic to the building. However, unlike the building, the LSV complex is a quotient with non-trivial cohomologies. The hardness comes from the fact that local views cannot capture the cohomology, which is a global property. Given this observation, the proof of completeness can be carried out following the argument of Ben-Sasson and Wigderson~\cite{BenSassonW2001} that any short resolution proof is narrow, and Grigoriev~\cite{Grigoriev2001} and Schoenebeck~\cite{Schoenebeck2008}'s transformation from resolution lower bounds to SoS lower bounds.

Technically, we rely on two very different types of expansion or isoperimetry. In our proof of completeness, we rely on an isoperimetric inequality called Gromov's filling inequality, that says that balls are essentially the objects with smallest boundary in any CAT(0) space (a class of spaces that includes both Euclidean spaces and the affine building). In our proof of soundness, we rely on the cosystolic expansion of the LSV complex, as proven by Evra and Kaufman~\cite{EvraK2016}, which implies that any non-trivial element in the cohomology has constant weight. Both of these statements are related to expansion, yet they are distinct from other notions of expansion used in previous SoS lower bounds.

\paragraph{Relation to previous SoS gap constructions.} All previous constructions of hard instances for SoS can be viewed in the vertex/edge-variable framework (typically vertex-variable). To the best of our knowledge, all known hard instances, proving inapproximability in the SoS hierarchy, are \emph{random} instances; either both the complex $X$ and the function $\beta$ are random, or just the function $\beta$ is random. Explicit hard instances for SoS are known in proof complexity (e.g., Tseitin tautologies on expanders), however, we do not know how to transform these hard instances into ones for inapproximability. The proof of SoS hardness of these random instances relies on very strong expansion of the underlying complex~\cite{Schoenebeck2008} or on certain pseudorandom properties~\cite{KothariMOW2017}, both of which are not yet known to be explicitly constructible.

In contrast, our instances are ``anti-random''. They are very structured and easily distinguishable from random instances. For example, all balls around a vertex up to some radius are identical and have very specific structure. Naturally, the typical  analysis that works for random instances cannot work here. For example, soundness for random instances is based on choosing a random $\beta$ and using a union-bound argument to show that with high probability, every solution violates nearly half of the constraints. In contrast, for us, a random $\beta$ is not a good choice because the local structure will quickly detect local contradictions, ruining the completeness altogether.
\paragraph{Open directions.} Our construction of explicit hard SoS instances based on HDXs begs several questions, some of which we discuss below.
\begin{description}
\item[Improved soundness] Our construction yields 3XOR hard instances which are at most $(1-\mu)$-satisfiable, owing to the cosystolic expansion of the underlying HDX (more precisely, $\Syst^2(X) \geq \mu$, see \cref{sec:simplicial} for the definition of $\Syst^2$). Coupled with reductions in the SoS hierarchy~\cite{Tulsiani2009}, this yields 3XOR hard instances which are at most $(1/2 +\eps)$-satisfiable for every $\epsilon \in (0,1)$. Can we obtain such a result directly from the HDX construction (bypassing reductions), say by constructing HDXs which satisfy $\Syst^2(X)\geq 1/2 - \eps$? In addition to maintaining the HDX structure, bypassing reductions would also allow for perfect completeness, which is lost while using NP-hardness reductions. 
\item[Fooling more levels of the SoS hierarchy] Our hard instances fool only $O(\sqrt{\log n})$ levels of the SoS hierarchy, as our argument is based on the injectivity radius of the complexes, which is $O(\log n)$, and we suffer a further square-root loss due to the use of Gromov's isoperimetry inequality. It is possible that a much stronger lower bound holds for these instances. Can one construct explicit hard instances that fool linearly many levels of the SoS hierarchy?
\item[HDX dimension and CSP definition] We find the contrast between the vertex-variable and edge-variable constructions baffling: while the vertex-variable construction is easy, our result demonstrates the hardness of the edge-variable construction. As we go to higher dimensions of HDX, there are more ways to define CSPs. Which of these are easy and which are hard?
\end{description}

\section{Preliminaries}
\subsection{The Sum-of-Squares hierarchy}
The sum-of-squares hierarchy\footnote{For more on Sum-of-Squares, see the recent monograph by Fleming, Kothari and Pitassi~\cite{FlemingKP2019}.}  provides a hierarchy of semidefinite programming (SDP) relaxations, for various combinatorial optimization problems. \Cref{fig:sos} describes the relaxation given by $t$ levels of the hierarchy for an instance $\mathcal{I}$ of 3XOR in $n$ variables, with $m$ constraints of the form $x_{i_1}+x_{i_2}+x_{i_3} = \beta_{i_1i_2i_3}$ over $\F_2$. We also use $\mathcal{I}$ to denote the set of all tuples $\{i_1,i_2, i_3\}$ present as constraints. A solution to the relaxation is specified by a collection of unit vectors $\left\{\U{S}\right\}_{S \subseteq [n], \abs{S} \leq t}$, satisfying the constraints in the program. The objective equals the fraction of constraints ``satisfied'' by the SDP solution.
\begin{figure}[ht]
\hrule
\vline
\begin{minipage}[t]{0.99\linewidth}
\vspace{-5 pt}
{\small
\begin{align*}
\mbox{maximize}\quad ~~
\frac12 + \frac{1}{2m} \cdot &\sum_{\{i_1,i_2,i_3\} \in \mathcal{I}} (-1)^{\beta_{i_1i_2i_3}} \cdot \ip{\U{\{i_1,i_2,i_3\}}, \Uempty} \\
\mbox{subject to}\quad
 \ip{\U{S_1}, \U{S_2}}
  &~=~ \ip{\U{S_3}, \U{S_4}}
  & \forall~ S_1 \Delta S_2 = S_3 \Delta S_4, \abs{S_1}, \ldots, \abs{S_4} \leq t \\
 \norm{\U{S}} &~=~ 1 & \forall S, \abs{S} \leq t
\end{align*}}
\vspace{-14 pt}
\end{minipage}
\hfill\vline
\hrule
\caption{Relaxation for 3XOR given by $t$ levels of the SoS hierarchy}
\label{fig:sos}
\end{figure}

To prove a lower bound on the value of the SDP relaxation, we will use the following result, which shows the existence of vectors $\U{S}$ yielding an objective value of 1, when the given system of XOR constraints does not have any ``low-width'' refutations. Formally, we consider a system called $\oplus$-resolution, where the only rule allows us to combine two equations $\ell_1=b_1$ and $\ell_2=b_2$ to derive the equation $\ell_1+\ell_2=b_1+b_2$. A refutation is a derivation of $0=1$. The width of a refutation is the maximum number of variables in any equation used in the refutation. We include a proof of the following lemma in \Cref{sec:sos-width}.
\begin{lemma}[{\cite[Lemma~13]{Schoenebeck2008}}, {\cite[Theorem~4.2]{Tulsiani2009}}]\label{lem:sos}
Let $\Lambda$ be a system of equations in $n$ variables over $\F_2$, which does not admit any refutations of width at most $2t$. Then there exist vectors $\left\{\U{S}\right\}_{S \subseteq [n], \abs{S} \leq t}$ satisfying the constraints in \Cref{fig:sos}, such that for all equations $\sum_{i \in T} x_i = b_T$ in $\Lambda$ with $\abs{T} \leq t$, we have $\ip{\U{T},\Uempty} = (-1)^{b_T}$.
\end{lemma}

\subsection{Simplicial complexes}\label{sec:simplicial}

A simplicial complex $X$ is a non-empty collection of sets (known as faces) which is closed downwards. The $i$-dimensional faces $X(i)$ are all sets of size $i+1$. The dimension of the complex is the maximal dimension of a face. Faces of that dimension are known as facets. Faces of dimensions $0,1,2,3$ are called vertices, edges, triangles, and tetrahedra, respectively.

Graphs are $1$-dimensional simplicial complexes.
The skeleton of a simplicial complex is the graph obtained by retaining only faces of dimension at most~$1$.

\paragraph{Links} Let $X$ be a $d$-dimensional simplicial complex. The link $X_s$ of a face $s \in X(i)$ is a simplicial complex of dimension $d-(i+1)$ given by $X_s(j) := \{ t : s \cup t \in X(j+i+1) \}$. In other words, $X_s$ contains all faces in $X$ which contain $s$, with $s$ itself removed.

\paragraph{Balls} Let $X$ be a simplicial complex. A ball of radius $r$ around a vertex $v$ is the subcomplex induced by all vertices at distance at most $r$ from $v$, as measured on the skeleton of $X$. That is, the subcomplex contains a face of $X$ if it contains all the vertices of the face.

\paragraph{Simplicial map} If $X$ and $Y$ are two simplicial complexes, then a simplicial map $\psi\colon Y \to X$ is a map from $Y(0)$ to $X(0)$ that maps faces to faces.

\paragraph{Chains}

Fix a $d$-dimensional simplicial complex $X$. Let $C^i = C^i(X,\F_2)$ be the set of all functions from $X(i)$ to $\F_2$. Elements of $C^i$ are also known as $i$-chains.

For an $i$-chain $f$, we define $|f|$ to be the number of non-zero elements in $f$. For two $i$-chains $f,g$, we define the distance between $f$ and $g$ to be $\dist(f,g) = |f+g|$.

\paragraph{Inner product}
For $f,f'\in C^i$, let us denote by $\ip{f,f'}_i$ the following sum modulo $2$:
\[\ip{f,f'}_i := \sum_{s\in X(i)}f(s)f'(s) .\]

This is not an inner product in the usual sense as we are working over a field of non-zero characteristic, but it is convenient notation. We will usually drop the subscript $i$. 

\paragraph{Dual space} Given any subspace $V \subset C^i$, the dual of $V$ (under $\ip{\cdot,\cdot}_i$) is defined as:
\[V^{\perp}:= \{ f \in C^i \mid \text{ for all } g \in V, \ip{f,g}_i = 0 \}.\]

\paragraph{Boundaries, Cycles, Homology}

The boundary operator $\dho_i\colon C^{i}\to C^{i-1}$ is given by
\[
 \dho_i f(s) := \sum_{t\in X(i)\colon t \supset  s} f(t).
\]
It gives rise to boundaries $B_i$ and cycles $Z_i$:
\[
 B_i := \im \dho_{i+1},\qquad Z_i:= \ker \dho_i.
\]
In the case of graphs, $Z_1$ consists of all sums of cycles (in the usual sense).

The coboundary operator $\delta_i\colon C^i\to C^{i+1}$, which is the adjoint of the boundary operator,  is given by
\[
 \delta_i f(t) := \sum_{s\in X(i) \colon s \subset t}f(s) = \sum_{u\in t} f(t-\set u).
\]
It gives rise to coboundaries and cocycles:
\[ B^i := \im \delta_{i-1},\qquad Z^i := \ker \delta_i. \]
We will usually drop the subscript $i$ when invoking $\dho,\delta$.

It is easy to see that $B_i\subset Z_i$ (every boundary is a cycle) and $B^i\subset Z^i$ (every coboundary is a cocycle). For example, in a $2$-dimensional complex, the boundary of every triangle is a cycle. We call such cycles \emph{trivial cycles}. Modding out by trivial cycles and cocycles, we obtain the homology and cohomology spaces
\[ H_i := Z_i/B_i,\qquad H^i := Z^i/B^i. \]
The dimensions of these spaces (which are identical) measure the number of ``holes'' in a particular dimension. Nice complexes (such as the buildings considered below) have no holes.

The following claim shows that that the coboundary operator is the adjoint of the boundary operator.
\begin{claim}\label{claim:dual}
Let $f\in C^i, g\in C^{i-1}$. Then $\ip{f,\delta g}_i=\ip{\dho f,g}_{i-1}$.
\end{claim}
\begin{proof}\begin{align*}
    \ip{f,\delta_{i-1} g}_i & = \sum_{t \in X(i)} f(t) \cdot \delta_{i-1} g(t) = \sum_{t \in X(i)} f(t) \cdot \left(\sum_{s \in X(i-1) \colon s\subset t} g(s)\right) \\
                      & =\sum_{s \in X(i-1)} \left( \sum_{t \in X(i) \colon t \supset s} f(t) \right) \cdot g(s) = \ip{\dho_i f,g}_{i-1}. \qedhere
                        \end{align*}
  \end{proof}

  The following claims shows the dimensions of homology and cohomology spaces are identical.
  \begin{claim}\label{claim:ZBduals}
    $Z_i = (B^i)^\perp, \qquad Z^i = (B_i)^\perp$.
  \end{claim}
  \begin{proof} $Z_i = \ker \partial_i = \ker \delta_{i-1}^* = (\im \delta_{i-1})^\perp = (B^i)^\perp$. \qedhere

%
    \end{proof}

\begin{claim}\label{claim:homeqcohom}
  $\dim H_i = \dim H^i$
\end{claim}
\begin{proof}\begin{align*}
    \dim H_i & = \dim Z_i - \dim B_i\\
             &= \dim C^i - \dim B^i - \dim B_i && [\text{By \cref{claim:ZBduals}}]\\
             & = \dim Z^i - \dim B^i && [\text{By \cref{claim:ZBduals}}]\\
             &= \dim H^i. && \hfill \qedhere
  \end{align*}
\end{proof}


\paragraph{Cosystoles} We define, following Evra and Kaufman~\cite[Definition 2.14]{EvraK2016}, the $i$-cosystole of a complex $X$ to be the minimal (fractional) size of $f\in Z^i\setminus B^i$,
\[ \Syst^i(X) := \min_{f\in Z^i\setminus B^i} |f|/|X(i)|.
\]



\subsection{The building $\B^{(d+1)}$}

The infinite $k$-regular tree is the unique connected $k$-regular graph without cycles. Affine buildings are higher-dimensional analogs of the infinite $k$-regular tree. For $d=1$, the one-dimensional affine building $\B^{(1)}$ is the $k$-regular tree. For higher dimensions
%
they are \emph{regular} in the sense that all vertex links are bounded and identical in structure, they are connected and contractible,\footnote{A complex is contractible, roughly speaking, if it can be continuously deformed to a point (technically, it is homotopy-equivalent to a point). Since (co)homologies are preserved by such deformations, all (co)homologies of a contractible complex vanish. } and so have vanishing cohomologies, that is, the cohomology spaces $H^1,\ldots,H^{d-1}$ are trivial, where $d$ is the dimension.


We won't describe $\B^{(d+1)}$ any further; the interested reader can check~\cite{Ji2012,AbramenkoBrown}. A crucial property of $\B^{(d+1)}$ which we will need in the sequel is its being a CAT(0) space,\footnote{A space is CAT(0) if for every triangle $x,y,z$, the distance between $x$ and the midpoint of $y,z$ is at most the corresponding distance in a congruent triangle in Euclidean space.} which is a geometric definition capturing non-positive curvature; see~\cite{BridsonHaefliger} for more information.
The property of being CAT(0) has the following implication, due to Gromov~\cite{Gromov1983,Guth2006,Wenger2008}: 
\begin{theorem}[Gromov's filling inequality for CAT(0) spaces]
For every cycle $f\in Z_1$ there is a filling $g\in C_2$ such that $f = \partial g $ and $|g| = O(|f|^2)$.
\end{theorem}

Gromov's filling inequality is an isoperimetric inequality. It generalizes the classic isoperimetric inequality in the plane, which states that any simple closed curve of length $L$ encloses a region whose area is at most $L^2/4\pi$.

The isoperimetric inequality in the plane can be stated in an equivalent way: the boundary of any \emph{bounded} region of area $A$ is a curve whose length is at least $\sqrt{4\pi A}$. This inequality fails for unbounded regions, which could have infinite area but finite boundary (for example, consider the complement of a circle). In the same way, Gromov's inequality doesn't imply that each $g \in C_2$ satisfies $|\partial g| = \Omega(\sqrt{|g|})$. Rather, we have to replace $|g|$ with $\min_{h \in Z_2} |g+h|$.

Gromov's filling inequality also applies to $i$-chains, with an exponent of $i+1$, but we will only need the case $i=1$.

\medskip

In the sequel, we will apply Gromov's filling inequality not to the building itself, but rather to balls in the building. The CAT(0) property almost immediately implies that a ball in a CAT(0) space is itself CAT(0)~\cite[Exercise II.1.6]{BridsonHaefliger}. Furthermore, it is well-known that CAT(0) spaces are contractible, and so have vanishing homologies.

\begin{lemma} \label{lem:balls-building}
Balls in $\B^{(d+1)}$ have vanishing homologies and satisfy Gromov's filling inequality.
\end{lemma}

\subsection{The LSV quotient}\label{sec:LSV} 
Whereas the affine building is an infinite simplicial complex, Lubotzky, Samuels and Vishne constructed a growing family of finite complexes that are obtained from quotients of the affine building. These quotients have a growing number of vertices, and locally, in a ball around each vertex, the complex is isomorphic to the affine building. Moreover, they gave a very explicit algorithm for constructing these complexes by first constructing a Cayley graph with an explicit set of generators, and then the higher dimensional faces are simply the cliques in the Cayley graph.

\begin{theorem}[Lubotzky, Samuels, Vishne~{\cite[Theorem 1.1]{LubotzkySV2005-exphdx}}]
Let $q$ be a prime power, $d\ge 2$. For every $e>1$ the group $G = PGL_d(\mathbb{F}_{p^e})$ has an (explicit) set of $\sqbinom{d}{1}_q + \sqbinom{d}{2}_q + \ldots + \sqbinom{d}{d-1}_q$ generators, such that the Cayley complex of $G$ with respect to these generators is a Ramanujan complex $X$ covered by $\B^{(d)}(F)$ for $F=\mathbb{F}_q(\!(y)\!)$.
\end{theorem}
The precise definition of ``Ramanujan complex'' is not important for this context. For us, there are three important aspects of this theorem: efficient construction, local structure, and global structure. 
\begin{itemize}
\item {\bf Efficient construction:} Firstly, the fact that the complex is constructible in polynomial time.
\item {\bf Local structure:} Next, we highlight the fact that locally the complex looks like the building. The theorem states that the complex $X$ is covered by $\B^{(d)}$. A {\em covering map} maps a simplicial complex $Y$ {\em surjectively} to a simplicial complex $X$ by mapping the vertices $\psi:Y(0)\to X(0)$ such that for every $k\leq d$ the image of every $k$-face $\set {v_0,\ldots,v_k}\in Y(k)$ is a $k$-face $\set{\psi(v_0),\ldots,\psi(v_k)}\in X(k)$.

The fact that $X$ is covered by $\B^{(d)}$ means that the neighborhood of a vertex in $X$ and in $\B^{(d)}$ look exactly the same. It turns out that for the LSV complexes this continues to be true also for balls of larger radius around any vertex.  This is a higher-dimensional analog of the graph property of containing no short cycles (locally looking like a tree). Define the injectivity radius of $X$ to be the largest $r$ such that the covering map $\B^{(d)}\stackrel{\psi}{\rightarrow} X$ is injective from balls of radius $\leq r$ in $\B^{(d)}$ and the ball of radius $\leq r$ in $X$. We do not mention the center of the ball they are all isomorphic.
\begin{theorem}[Lubotzky and Meshulam \cite{LubotzkyM2007}, see also\footnote{The theorem was proven by \cite{LubotzkyM2007}. They stated their theorem using a slightly different definition for injectivity radius but one can prove that the two definitions coincide in this case. This was reproven in \cite{EvraGL2015} who use the definition of injectivity radius that is convenient for us.}
 {\cite[Corollary 5.2]{EvraGL2015}}]
Let $X$ be the LSV complex above. Then the injectivity radius $r(X)$ of $X$ satisfies 
\[ r(X) \ge \frac{\log_q|X|}{2(d-1)(d^2-1)}-\frac 1 2
\] where $|X|$ is the number of vertices in $X$.
\end{theorem}
\item {\bf Global structure:} Finally, we look at the second cohomology group of the LSV complexes. Kaufman, Kazhdan and Lubotzky~\cite{KaufmanKL2016} showed that the groups defining the LSV quotient complexes can be chosen so that the second homology is non-empty. 

\begin{proposition}[Kaufman, Kazhdan, Lubotzky~{\cite[Proposition 3.6]{KaufmanKL2016}}]\label{prop:KKL}
There is an infinite and explicit sequence of LSV complexes with a non-vanishing second cohomology.\end{proposition}
We remark that Kaufman, Kazhdan and Lubotzky~\cite{KaufmanKL2016} proved that these complexes exist. To show that they are also efficiently constructible, we look into their proof to recall the construction: start with any LSV complex $X$ viewed as a Cayley graph of a group $G$. Find some element of order $2$ in $G$ (such an element always exists), and then quotient $X$ by this element, thus obtaining a complex $Y$ that is itself is a Ramanujan complex because it is a quotient of one. $Y$ is clearly efficiently constructible from $X$, and has half as any vertices. This construction shows (see \cite[Proposition 3.5]{KaufmanKL2016}) that $H^1(Y)\neq 0$. Furthermore, the proof of \cite[Proposition 3.6]{KaufmanKL2016} shows that because $G$ has property $T$ one can deduce also that $H^2(Y)\neq 0$.

Evra and Kaufman proved~\cite[Theorem 1.9]{EvraK2016} that quotients of $\B^{(d)}$ (and even a more general class of complexes) are so-called ``cosystolic expanders'' which in particular implies the following.
\begin{theorem}[Evra and Kaufman~{\cite[Part of Theorem 1.9]{EvraK2016}}]
Let $\set{X_n}$ be a family of LSV complexes. There exists some constant $\mu>0$ that depends only on $q$ and $d$ but not on the size $n$ of the complex, such that every $f\in Z^2(X)\setminus B^2(X)$ must have weight at least $\mu\cdot |X(2)|$. \end{theorem}
\end{itemize}
 \section{Main result} \label{sec:main}

\subsection{Local geometry of LSV complexes}\label{sec:geometry}
%
%
The infinite sequence of complexes we will be working with are the LSV complexes described in Section \ref{sec:LSV} above. The properties we care about are (1) that they are efficiently constructible, (2) that small balls in these complexes are isomorphic to the affine building, which satisfies certain isoperimetric inequalities because it is a CAT(0) space, and (3) that each complex has a two-dimensional cocycle with linear distance from the set of coboundaries. The second and third properties provide the tension between the local and the global structure of these complexes that we now harness for our hardness. 
\medskip

To construct an SDP solution, we will need to show that our instance based on the LSV complex ``locally looks satisfiable''. To this end, we will first develop some local properties of the LSV complex.

Note that each $h \in C^2$ corresponds to a set of triangles. For the following statements, we consider two triangles to be connected if they share an edge. This can be used to define connected components. Note that if $h$ can be split into connected components $h_1, \ldots, h_s$, then the components correspond to \emph{disjoint} sets of triangles. Moreover, no triangle in $h_i$ shares an edge with a triangle in $h_j$ when $i \neq j$, which also implies that the boundaries $\dho h_i$ and $\dho h_j$ correspond to disjoint sets of edges.

We prove the following claims by mapping small connected sets in $X(2)$ to corresponding sets in the infinite building $\B$.
%
%
%
%
The first proposition shows that there can be no small non-trivial cancellations (i.e., not coming from tetrahedra).
\begin{proposition}\label{prop:boundary}
Let $h_0 \in C^2$ be a connected set of triangles such that $|h_0| < r$ and $\dho h_0 = 0$. Then $h_0 \in B_2$.
\end{proposition}
\begin{proof}
Since $|h_0|<r$, there is a ball $\ballX$ of radius $r$ that contains the support of $h_0$. By assumption, the covering map $\psi\colon \B\to X$ has injectivity radius of at least $r$. This means that there is a radius-$r$ ball $\ballB = \psi^{-1}(\ballX)$ in $\B$ that is isomorphically mapped by $\psi$ to $\ballX$. Look at $\hat h_0 = \psi^{-1}(h_0) \in C^2(\ballB)$, the chain isomorphic to $h_0$ in the building. Clearly $\dho \hat h_0 = \psi^{-1} (\dho h_0) = 0$, and since balls in the building have zero homologies by \Cref{lem:balls-building}, we deduce that $\hat h_0$ itself must be a boundary, i.e.\ there must be some $\hat g_0\in C^3(\ballB)$ such that $\dho \hat g_0 = \hat h_0$. Moving back to $X$, we see that $g_0 := \psi(\hat g_0)\in C^3(X)$ necessarily satisfies $\dho g_0 = h_0$, and so $h_0 \in B_2$.
\end{proof}
This proposition states that locally (i.e., within the injective radius $r$), $Z_2$ looks like $B_2$. We thus have a complex whose cohomology group is non-trivial, yet locally, the \emph{homology} group ``looks'' trivial. Note that this is a twist on what we had claimed in the introduction, a complex whose cohomology group is non-trivial, yet locally, the \emph{cohomology} group ``looks'' trivial. However, these are identical statements owing to \cref{claim:homeqcohom}.

The next proposition shows that Gromov's filling inequality in the infinite building $\B$ can be used to yield a similar consequence for small sets in the finite complex $X$.
\begin{proposition}\label{prop:filling}
Let $h_0 \in C_2$ be a connected set of triangles such that $|h_0| < r$ and $|h_0| \leq |h_0+h|$ for all $h \in B_2$. Then, $|\dho h_0| \geq c \cdot |h_0|^{1/2}$, where $c > 0$ is an absolute constant.
\end{proposition}
\begin{proof}
  As before, the support of $h_0$ is contained in a ball $\ballX$ of $X$ which is isomorphic under $\psi$ to a ball $\ballB$ in $\B$. Let $\hat h_0 = \psi^{-1}(h_0)\in C^2(\B)$, and let $\hat f_0 = \dho \hat h_0$. We now apply the filling theorem of Gromov, which holds in $\ballB$ due to \Cref{lem:balls-building}, to deduce that there is some $\hat h_1$ that fills $\hat f_0$, namely $\dho \hat h_1 = \hat f_0$, and whose size is at most $|\hat h_1| = O(|\hat f_0|^2)$.

Now $\dho(\hat h_0 - \hat h_1) = \hat f_0 - \hat f_0 = 0$. Since the ball $\ballB$ has zero homologies by \Cref{lem:balls-building}, $\hat h_0-\hat h_1$ itself must be a boundary: there must be some $\hat g\in C^3(\ballB)$ such that $\dho \hat g = \hat h_0-\hat h_1$. Pushing $\hat g$ and $\hat h_1$ back to $X$, we get $g = \psi(\hat g)$ and $h_1 = \psi(\hat h_1)$, which satisfy $\dho g = h_0 - h_1$. At this point we have a \emph{small} $h_1$ 
that is close via a boundary to $h_0$. Finally, observe that $f_0=\dho h_0$ satisfies $f_0 = \psi^{-1}(\hat f_0)$. So
\[
 |f_0| ~=~ |\hat f_0| ~\ge~ c \cdot |\hat h_1|^{1/2} ~=~  c \cdot |h_1|^{1/2} ~\ge~ c \cdot |h_0|^{1/2}, \]
where the last inequality used that $|h_0| \leq |h_0 + (h_1-h_0)|$, since $h_1-h_0 = \partial g \in B_2$.
\end{proof}

\subsection{Fooling $\Omega(\sqrt {\log n})$ levels of SoS hierarchy}
Let $X$ be a $d$-dimensional LSV complex, with $\abs{X(1)} = n$ and non-trivial second cohomology group, as per Proposition \ref{prop:KKL}. Below, we construct an instance of 3XOR in $n$ variables using this complex, and prove a lower bound on the integrality gap of the relaxation obtained by $\Omega(\sqrt{\log n})$ levels of the SoS hierarchy.
\paragraph{Construction.} We construct a system of equations on $X$ by putting a variable $x_{\{a,b\}}$ for each edge $\{a,b\}\in X(1)$ of the complex, and an equation \[x_{\{a,b\}}+x_{\{b,c\}}+x_{\{c,a\}}=\beta_{\{a,b,c\}}\] for each triangle $\{a,b,c\}\in X(2)$, where $\beta$ is an arbitrary element of $Z^2\setminus B^2$.

Recall that $X$ can be constructed efficiently. Given $X$, we can find a vector $\beta \in Z^2 \setminus B^2$ using elementary linear algebra. Therefore the entire system can be constructed efficiently.

\paragraph{Soundness.} Soundness of this system follows easily from the fact that the cosystole is large.
\begin{claim}[Soundness]\label{claim:soundness}
Every assignment to the system defined above falsifies at least $\mu$ fraction of the equations.
\end{claim}
\begin{proof}
An assignment to the variables is equivalent to an $f\in C^1$. Every equation satisfied by $f$ is a triangle in which $\delta f(\{a,b,c\}) = \beta_{\{a,b,c\}}$, and so the number of unsatisfied equations is $\dist(\delta f, \beta) = |\delta f+\beta|$. Since $\delta f \in B^2$ and $\beta \in Z^2 \setminus B^2$, also $\delta f + \beta \in Z^2 \setminus B^2$, and so $|\delta f + \beta|/|X(2)| \geq \Syst^2(X) \geq \mu$. In other words, the assignment falsifies at least a $\mu$ fraction of the equations.
\end{proof}

The main work is to prove completeness, namely to show that the system looks locally satisfiable.

\paragraph{Completeness.} Our main result is that this system appears satisfiable to the Sum-of-Squares hierarchy with $O(\sqrt{\log n})$ levels. Grigoriev~\cite{Grigoriev2001} and Schoenebeck~\cite{Schoenebeck2008} showed that to prove such a statement it suffices to analyze the \emph{refutation width} of the system of equations (see \Cref{lem:sos}). If the refutation width is at least $w$, then $w/2$ levels of the Sum-of-Squares hierarchy cannot refute the system.

A system of linear equations over $\mathbb{F}_2$ can be \emph{refuted} using a proof system known as \emph{$\oplus$-resolution}, in which the only inference rule is: given $\ell_1=b_1$ and $\ell_2=b_2$, deduce $\ell_1 \oplus \ell_2 = b_1 \oplus b_2$; here $\ell_1,\ell_2$ are XORs of variables, and $b_1,b_2$ are constants. A refutation has the structure of a directed acyclic graph (DAG) where each non-leaf node has two incoming edges.
A \emph{refutation} is a derivation which starts with the given linear equations, placed at the leaves of a DAG, and reaches the equation $0=1$ at the root of the DAG. The \emph{width} of a linear equation $\ell = b$ is the number of variables appearing in $\ell$. The width of a refutation is the maximum width of an equation in any of the nodes of the DAG. 
\smallskip

In the remainder of this section, we prove the following theorem, which together with \cref{lem:sos} implies \cref{thm:main-intro}.
\begin{theorem} \label{thm:width-lb}
The construction above requires width at least $\Omega(\sqrt{r})$ to refute in $\oplus$-resolution, where $r = \Theta(\log n)$ is the injectivity radius of the complex.
\end{theorem}


The proof follows classical arguments of Ben-Sasson and Wigderson~\cite{BenSassonW2001} regarding lower bounds on resolution width, which were also used in the proof of Schoene\-beck~\cite{Schoenebeck2008}. Whereas Ben-Sasson and Wigderson relied on boundary expansion, we rely on Gromov's filling inequality (and so lose a square root).

Suppose we are given a refutation for this system, and consider the corresponding DAG. Each leaf $\nu$ in the DAG is labeled by a triangle $T_\nu\in X(2)$. Define
\[h_\nu := \one_{T_\nu} \in C^2, \quad \and \quad b_\nu := \beta_{T_\nu}\in \F_2.\]
For each inner node $\nu$ in the DAG, let $\nu_1,\nu_2$ be its two incoming nodes. Define inductively,
\[ h_\nu := h_{\nu_1}+h_{\nu_2} \in C^2, \quad \and \quad b_\nu := b_{\nu_1}+b_{\nu_2}\in \F_2  .\]
\begin{proposition}
For every node $\nu$, $b_\nu = \ip {\beta,h_\nu}$.
\end{proposition}
\begin{proof}
This is immediate by following inductively the structure of the DAG.
\end{proof}
As in~\cite{BenSassonW2001}, we next define a complexity measure for each node of the DAG.
While in~\cite{BenSassonW2001} the complexity measure is based on the number of ``leaf equations'' used to derive the one at a given node, we will need to discount sets of triangles corresponding to tetrahedra, as these cannot lead to contradictions.
Recall that $B_2 = \im \dho_3$ is the set of triangle chains that ``come from'' tetrahedra chains, which we consider as the ``trivial'' cycles. We define a complexity measure at each node,
\[ \kappa(\nu) := \dist(h_\nu,B_2) = \min_{h\in B_2} | h_\nu + h|\]
that measures the distance of $h_\nu$ from these trivial cycles.
The complexity measure $\kappa$ satisfies the following sub-additivity property.
\begin{proposition}\label{prop:subadditive}
If $\nu$ is an inner node in the DAG with $\nu_1, \nu_2$ its two incoming nodes, then
\[ \kappa(\nu) \leq \kappa(\nu_1) + \kappa(\nu_2) . \]
\end{proposition}
\begin{proof}
  Let $h_1, h_2 \in B_2$ be such that $\kappa(\nu_1) = |h_{\nu_1}+h_1|$ and $\kappa(\nu_2) = |h_{\nu_2}+h_2|$. Recall that $h_\nu = h_{\nu_1}+h_{\nu_2}$. Then, we have
 \begin{align*}
    \kappa(\nu_1) + \kappa(\nu_2) ~=~ |h_{\nu_1}+h_1| + |h_{\nu_2}+h_2| &~\ge~ |h_{\nu_1}+h_{\nu_2}+h_1+h_2| \\
    &~=~ |h_\nu + h_1 + h_2| ~\ge~ \kappa(\nu). \qedhere
 \end{align*}
\end{proof}
We also need the fact that the complexity of a node with a contradiction must be non-zero.
\begin{proposition}\label{prop:nonzero}
If $\kappa(\nu)=0$ then $b_\nu=0$.
\end{proposition}
\begin{proof}
  If $\kappa(\nu)=0$ then $h_\nu\in B_2 $. Hence $b_\nu = \ip{\beta, h_\nu} = 0$ since $\beta \in Z^2 = (B_2)^\perp$ (\cref{claim:ZBduals}).
\end{proof}
Next, we consider the width of each node in the DAG.
For a node $\nu$, let
\[ f_\nu := \dho h_\nu \in C^1 .\]
Thus $f_\nu$ indicates the set of variables appearing in the left-hand side of the equation on node $\nu$. So the width of the system is the maximum, over all nodes $\nu$ in the DAG, of $|f_\nu|$.

We can now prove \Cref{thm:width-lb} using the above complexity measure, and results from \Cref{sec:geometry}.
\begin{proof}[Proof of \Cref{thm:width-lb}]
Let $\nu^*$ denote the root of the DAG. By virtue of being a refutation, $b_{\nu^*} = 1$ while $f_{\nu^*}=0$.
In other words, $\dho h_{\nu^*} = f_{\nu^*}=0$, which means that $h_{\nu^*}\in Z_2$. Since $b_{\nu^*}=1$, we also have by \Cref{prop:nonzero} that $\kappa(\nu^*) > 0$.

Let $h \in B_2$ be such that $\kappa(\nu^*) = |h_{\nu^*}+h|$, and let $h_1, \ldots, h_s$ be the disjoint connected components of $h_{\nu^*}+h$.
We will first show that $\kappa(\nu^*) = |h_{\nu^*}+h|\geq r$. Assuming $\kappa(\nu^*)<r$, we have that
\[
|h_1| + \cdots + |h_s| ~=~ |h_{\nu^*}+h| ~<~ r .
\]
Also, since
\[
\dho h_1 + \cdots + \dho h_s ~=~ \dho(h_{\nu^*}+h) ~=~ \dho h_{\nu^*} ~=~ 0 ,
\]
we must have that $\dho h_i = 0$ for each $i \in [s]$, since connected components have disjoint boundaries. Applying \Cref{prop:boundary} to each $h_i$, we get that $h_i \in B_2$ for each $i \in [s]$. However, this implies $h_{\nu^*}+h \in B_2$ and hence $\kappa(\nu^*) = 0$, which is a contradiction.

Using sub-additivity (\Cref{prop:subadditive}), $\kappa(\nu^*) \ge r$,  and
the fact that the leaves  of the DAG satisfy $\kappa(\nu) = 1$, we get that there must be some
internal node $\nu$ for which
$r/2 \le \kappa(\nu)<r$. We can find such a node by starting at the root and always going to the child with higher complexity, until reaching a node $\nu$ such that $\kappa(\nu) < r$. We will prove that for such a node, we must have $|f_\nu| = \Omega(\sqrt{r})$. 

As before, let $h \in B_2$ now be such that $\kappa(\nu) = |h_\nu+h|$, and let $h_1, \ldots, h_s$ be the disjoint connected components of $h_\nu+h$. We have that $|h_i| \leq |h_\nu+h| < r$ for each $i \in [s]$. By the minimality of $|h_\nu+h|$, we also have that for any $h' \in B_2$ and any $i \in [s]$,
\[
|h_i| +  |h_\nu + h - h_i| ~=~ |h_\nu+h| ~\leq~ |h_\nu+h + h'| ~\leq~ |h_i+h'| + |h_\nu + h - h_i|  .
\]
Thus, $|h_i|$ is also minimal for each $i$, and we can apply \Cref{prop:filling} to each connected component $h_i$, to obtain
\begin{align*}
  |f_\nu| ~=~ |\dho(h_\nu+h)| ~=~ |\dho h_1| + \cdots + |\dho h_s| &~\ge~ c \cdot |h_1|^{1/2} +\cdots + c \cdot |h_s|^{1/2} \\
                                                               &~\ge~ c\cdot (|h_1|+ \cdots + |h_s|)^{1/2} \\
                                                               &~=~ c\cdot |h_\nu+h|^{1/2} ~\geq~ (c/\sqrt{2}) \cdot \sqrt{r}. \qedhere
\end{align*}

\end{proof}

\section*{Acknowledgements}
Part of this work was done when the authors were visiting the Simons Institute of Theory of Computing, Berkeley for the 2020 summer cluster on "Error-Correcting Codes and High-Dimensional Expansion". We thank the Simons institute for their kind hospitality.

{\small
\bibliographystyle{prahladhurl}
\bibliography{sosgap-bib.bib}
}

\appendix

\section{Proof of \Cref{lem:sos}}\label{sec:sos-width}
\noindent {\bf \cref{lem:sos} (Restated)} (\cite[Lemma~13]{Schoenebeck2008}, \cite[Theorem~4.2]{Tulsiani2009}) {\em 
Let $\Lambda$ be a system of equations in $n$ variables over $\F_2$, which does not admit any refutations of width at most $2t$. Then there exist vectors $\left\{\U{S}\right\}_{S \subseteq [n], \abs{S} \leq t}$ satisfying the constraints in \Cref{fig:sos}, such that for all equations $\sum_{i \in T} x_i = b_T$ in $\Lambda$ with $\abs{T} \leq t$, we have $\ip{\U{T},\Uempty} = (-1)^{b_T}$.
}

%
\begin{proof}
We assume that $\Lambda$ is closed under width-$2t$ $\oplus$-resolution, replacing $\Lambda$ by its closure if necessary, and also that it contains the trivial equation $0=0$. We will now construct the unit vector $\U{S}$.

Define a relation $\sim$ on subsets of $[n]$ of size at most $t$ as follows: $S \sim T$ iff there exists an equation $\sum_{i \in S \Delta T} x_i = b$ in $\Lambda$ for some $b \in \F_2$. It is easy to check that the relation is reflexive and symmetric. It is also transitive since for $S_1 \sim S_2$, $S_2 \sim S_3$, we can add the corresponding equations to obtain one of the form $\sum_{i \in S_1 \Delta S_3} x_i = b$ for some $b \in F_2$. Since $\abs{S_1}, \abs{S_3} \leq t$, this equation has at most $2t$ variables and must be in $\Lambda$ by the closure property. Thus, we have an equivalence relation which partitions all sets of size at most $t$ into equivalence classes, say $\C_1, \ldots, \C_s$. Choose an arbitrary representative $R_i$ for each class $\C_i$, and let $R(S)$ denote the representative for the class containing $S$. For convenience, we choose $R(\emptyset) = \emptyset$.

  We now construct the SDP vectors. Let $e_1, \ldots, e_s$ be an arbitrary orthonormal set of vectors, and assign $\U{R_i} = e_i$ for all $i \in [s]$. Note that for any $S$ with $\abs{S} \leq t$, there must be a \emph{unique} equation of the form $\sum_{i \in S \Delta R(S)} x_i = b_S$ in $\Lambda$, since two different equations can be used to obtain a width-$2t$ refutation. We assign the vector for $S$ as
\[
\U{S} ~:=~ (-1)^{b_S} \cdot \U{R(S)} .
\]
The vectors are unit-length by construction. Note that if $S_1\Delta S_2 = S_3 \Delta S_4$, we must have $S_1 \sim S_2 \Leftrightarrow S_3 \sim S_4$. If $S_1 \not \sim S_2$, then we have that $\ip{\U{S_1}, \U{S_2}} = \ip{\U{S_3}, \U{S_4}} = 0$. Otherwise, we have $R(S_1)=R(S_2)$,  $R(S_3) = R(S_4)$, and equations of the form
\[
\sum_{i \in S_j \Delta R({S_j})} x_i = b_{S_j}\, , \quad j \in \{1,2,3,4\} .
\]
We must also have $b_{S_1}+b_{S_2} = b_{S_3} + b_{S_4}$, since otherwise we obtain two different equations with variables in $S_1 \Delta S_2 = S_3 \Delta S_4$, yielding a refutation. This suffices to satisfy the SDP constraints, since
\[
\ip{\U{S_1},\U{S_2}} ~=~ (-1)^{b_{S_1}+b_{S_2}} \cdot \ip{\U{R(S_1)}, \U{R(S_2)}} ~=~ (-1)^{b_{S_1}+b_{S_2}} ~=~ (-1)^{b_{S_3}+b_{S_4}} ~=~  \ip{\U{S_3},\U{S_4}} .
\]
Finally, for any equation $\sum_{i \in T} x_i = b_T$ in $\Lambda$ with $\abs{T} \leq t$, we get $\ip{\U{T}, \Uempty} = (-1)^{b_T}$, since we must have $T \sim \emptyset$ and $R(T) = \emptyset$.
\end{proof}

\end{document}